\documentclass[graybox]{svmult}

\usepackage{mathptmx}      
\usepackage{helvet}      
\usepackage{courier}       

\usepackage{makeidx}       
\usepackage{graphicx}    
                             
\usepackage{multicol}  
\usepackage[bottom]{footmisc}

\usepackage{url}

\usepackage{subeqnarray}
\usepackage{amsmath}
\usepackage{amsfonts}
\usepackage{bbm}

\makeindex

\begin{document}

\title*{Beyond adiabatic elimination: Effective Hamiltonians and singular perturbation}
\titlerunning{Beyond Adiabatic Elimination: Effective Hamiltonians and singular perturbation}

\author{Mikel Sanz, Enrique Solano, and \'I\~nigo L. Egusquiza}
\authorrunning{M. Sanz {\it et al.}}

\institute{Mikel Sanz \at Department of  Physical Chemistry, University of the Basque Country UPV/EHU, Apartado 644, E-48080 Bilbao, Spain, \email{mikel.sanz@ehu.eus} \and Enrique Solano \at Department of Physical Chemistry, University of the Basque Country UPV/EHU, Apartado 644, E-48080 Bilbao, Spain, \at IKERBASQUE, Basque Foundation for Science, Maria Diaz de Haro 3, 48013 Bilbao, Spain, \email{enr.solano@gmail.com} \and \'I\~nigo L. Egusquiza   \at Department of Theoretical Physics and History of Science, University of the Basque Country UPV/EHU, Apartado 644, E-48080 Bilbao, Spain,  \email{inigo.egusquiza@ehu.eus}}

\maketitle

\abstract{Adiabatic elimination is a standard tool in quantum optics, that produces an effective Hamiltonian for a relevant subspace of states, incorporating effects of its coupling to states with much higher unperturbed energy. It shares with techniques from other fields the emphasis on the existence of widely separated scales. Given this fact, the question arises whether it is feasible to improve on the adiabatic approximation, similarly to some of those other approaches. A number of authors have addressed the issue from the quantum optics/atomic physics perspective, and have run into the issue of non-hermiticity of the  effective Hamiltonian improved beyond the adiabatic approximation. Even though non-hermitian Hamiltonians are interesting in their own right, this poses conceptual and practical problems. 
\\
Here, we  first briefly survey methods present in the physics literature. Next we rewrite the problems addressed by the adiabatic elimination technique to make apparent the fact that they are singular perturbation problems from the point of view of dynamical systems. We apply the invariant manifold method for singular perturbation problems to this case, and show that this method produces the equation named after Bloch in nuclear physics. Given the wide separation of scales, it becomes intuitive that the Bloch equation admits iterative/perturbative solutions. We show, using a fixed point theorem, that indeed the iteration converges to a perturbative solution that produces in turn an exact Hamiltonian for the relevant subspace. We propose thus several sequences of effective Hamiltonians, starting with the adiabatic elimination and improving on it. We show the origin of the non-hermiticity, and that it is inessential given the isospectrality of the effective non-hermitian operator and a corresponding effective hermitian operator, which we build. We propose an application of the introduced techniques to periodic Hamiltonians.}

\section{Adiabatic Elimination} \label{sec:1}

Situations in which there is a wide separation among energy or time scales present in a physical system pervade all of Physics. It is therefore of recurrent interest to develop techniques to obtain approximate, effective descriptions of the low energy or slow sector of the system, since they are the most likely accessible to experimentation and control. In general, a na\"{\i}ve perturbation expansion will not provide us with the required effective description; depending on the scheme and approach, the obstacle to do so will turn up as secular terms in a time evolution, or as zero denominators in state expansions. Therefore, the approximation schemes valid for these situations will not be directly perturbative; they will be asymptotic, or resummation-based, or cumulant, or combinations thereof. Often, we have heuristic arguments for the construction of an effective description of the system, which do not naturally lead to improvements, at least not systematically.

A very common and useful approximation in quantum optics, normally based on a heuristic argument, is the so-called \emph{adiabatic elimination} technique \cite{walls2008quantum,Shore:2010uq,Yoo1985239}.  There are several ways of introducing this approximation. For this first presentation, consider the Schr\"odinger equation as a dynamical system, evolving with a Hamiltonian \( H \) as \( i\partial_t\psi=H\psi \) (we set here and henceforth \( \hbar=1 \)). Let state \( \psi \)Ê be partitioned into \( \alpha=P\psi \)Ê and \( \gamma=Q\psi \), with \( P \)  and \( Q =1-P\)   projectors, in such a manner that the eigenvalues of \( PHP \) are widely separated from those of \( QHQ \). Let it be the case in which the coupling between the \( P \) and \( Q \) subspaces is very small when compared to the eigenvalues of \( QHQ \).  To be more specific, let \( \tau \) be a characteristic scale, for instance, the norm of the restricted inverse \(\left(QHQ\right)^{-1}  \). Then, the Schr\"odinger equation may be rewritten as
\begin{subeqnarray}\label{eq:gensystem}
\I\partial_t \alpha&=& PHP \alpha+ PHQ \gamma\;,\\
\I\partial_t \gamma&=& QHP \alpha+ QHQ \gamma,
\end{subeqnarray}
where the last line of the system is to be multiplied by \( \tau \). Heuristically, we are asking that \( \tau QHP\ll 1 \), while \( \tau QHQ= O(1) \). This can be achieved if \( \gamma \) is small, and if we neglect its time evolution. In other words, we slave \( \gamma  \) to \(\alpha\) in the approximation \( QHQ \gamma=- QHP\alpha \), or, formally, \( \gamma= -(QHQ)^{-1}QHP\alpha \). By substituting this approximation for \(\gamma\) in the first component of the system, we obtain the effective evolution for the \emph{slow sector} \(\alpha\) as
\begin{equation}
\I\partial_t \alpha= \left(PHP- PHQ \frac{1}{QHQ} QHP\right)\alpha\;.
\end{equation}
By eliminating the fast component \(\gamma\), we obtain an effective Hamiltonian for the slow sector
\begin{equation}
H_{ \mathrm{adiabatic}}=PHP- PHQ \frac{1}{QHQ} QHP\;.
\end{equation}
This adiabatic elimination process has proven to be extremely useful in isolating the effective slow evolution for a number of systems . However, the aforementioned heuristic presentation does not lend itself readily to a systematic improvement. Thus, the question of how to transform this approximation into a controlled expansion is recurrently posed in quantum optics literature \cite{1751-8121-40-5-011,refId0}. In essence, both papers present an expansion of an integral kernel, either in the energy domain~\cite{1751-8121-40-5-011} or the time domain \cite{refId0}. 

As a matter of fact, this problem of finding a systematic expansion when there is a wide divergence in energy scales has also appeared in many other contexts: nuclear physics, condensed matter, and atomic and molecular physics, among others. In each case, there are expansion techniques, under different guises and names: Bloch equation~\cite{Eden1955,Ellis1977}, Schrieffer--Wolff expansion~\cite{SchriefferWolff1966}, or Born--Oppenheimer approximation~\cite{BornOppenheimer1927}. In this paper, we want to merge these techniques by presenting the system of equations as a singular perturbation problem, and by using the invariant manifold scheme for resummation of secular terms. For this purpose, we shall derive systematic expansions which improve on the adiabatic elimination technique and furnish us with an effective Hamiltonian for the low energy sector. Furthermore, we show the equivalence of this expansion and the Schrieffer--Wolff method. We apply this method to some simple examples in quantum optics, and suggest its future applicability for time-dependent periodic Hamiltonians.

\section{A singular perturbation problem and Bloch's equations}

In Physics, there are often problems depending on a small parameter which can be solved by applying perturbation theory. The singular perturbation theory approach has been developed for those cases for which no uniform regular expansion is possible, the so-called {\it singular problems}. Heuristically, one can identify a problem as being singular when it is qualitatively different at the zero order of the expansion, and immediately out of it. For example, the algebraic problem \( \epsilon x^2+ x-1=0 \) is singular for the small parameter \(\epsilon\), since it is second order if \(\epsilon\neq0\), but first order if \(\epsilon=0\). Similarly, singular perturbation problems in systems of ordinary differential equations \cite{Fenichel197953} appear frequently by reducing the order of the differential equation when the small parameter is set to be 0.
\begin{figure}[t] 
\centering\includegraphics[width=4.0in]{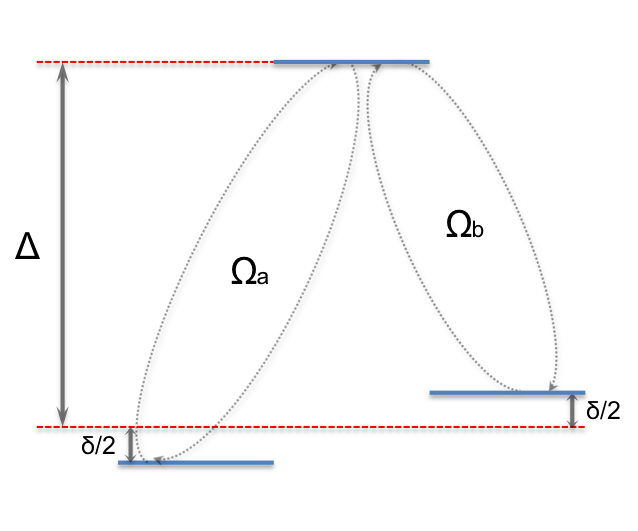}
\begin{center}\caption{The \(\Lambda\) system is a paradigmatic model in atomic physics and quantum optics, where a three-level atomic system may be coupled to classical or quantum fields.}\label{fig:lambda}\end{center}
\end{figure}

Let us now consider the \(\Lambda\) system depicted in Fig. \ref{fig:lambda}, which is ubiquitous in atomic physics. This is governed by the system of ordinary differential equations 
\begin{eqnarray}\label{eq:brionlambda}
\I\dot\alpha&=& - \frac{ \delta}{2} \alpha+ \frac{ \tilde{\Omega}^*_a}{2} \gamma\nonumber\,,\\
\I\dot \beta&=&  \frac{ \delta}{2} \beta+ \frac{ \tilde{\Omega}^*_b}{2} \gamma\,,\\
\I\dot \gamma&=& \frac{ \tilde{\Omega}_a}{2} \alpha +\frac{ \tilde{\Omega}_b}{2} \beta+ \Delta \gamma\,.\nonumber
\end{eqnarray}
We are interested in the regime in which \(\Delta\gg \delta, |\Omega_i|\), with \(\delta\) real. Under the change of variable \( t\to \delta t \) and the definitions \( \Omega_i=\tilde{\Omega}_i/\delta \)  and \( \epsilon= \delta/ \Delta \), the system transforms into
\begin{eqnarray}
\I\dot\alpha&=& - \frac{1}{2} \alpha + \frac{ \Omega^*_a}{2}\gamma\;,\nonumber \\
\I\dot \beta&=&  \frac{1}{2} \beta + \frac{ \Omega^*_b}{2}\gamma\;,\\
\I \epsilon \dot\gamma&=&  \gamma + \frac{ \epsilon}{2}\left( \Omega_a \alpha+ \Omega_b \beta\right)\,. \nonumber 
\end{eqnarray}
This is heuristically identifiable as a singular problem, since the last differential equation changes its character into an algebraic one as \(\epsilon \rightarrow 0\). Indeed, if one were to attempt a na\"{\i}ve perturbation expansion in the small parameter \(\epsilon\), one would obtain secular terms already at the first order, rendering the expansion invalid already for times $t$ of the order of \(  \Delta/ |\tilde{ \Omega}_i|^2 \) (in terms of the original variables).

The issue at hand is therefore how to identify a time-uniform scheme which provides us with approximations to the slowly varying variables. Among the numerous existing proposals, we shall concentrate here on the invariant manifold method~\cite{Fenichel197953}. This consists in constructing, among the submanifolds invariant under the flow, those that are perturbative in the small parameter. In the case of interest, as the system is linear, it is pertinent to only examine linear subspaces.

Let us write system (\ref{eq:gensystem}) in the matrix form
\begin{equation}\label{eq:newnotation}
\I\partial_t 
\begin{pmatrix}
\alpha \\ 
\gamma
\end{pmatrix}
= 
\begin{pmatrix}
\omega & \Omega^{\dag}\\ \Omega & \Delta
\end{pmatrix}
\begin{pmatrix}
\alpha\\ \gamma
\end{pmatrix}\,.
\end{equation}
For clarity, let us notice that if we were to write the \(\Lambda\) system in this more compact notation, we would be  using the isomorphism  \( \mathbb{C}^3=\mathbb{C}^2\oplus\mathbb{C} \). The three component vector \( \left( \alpha, \beta, \gamma\right)^T \) would be written  (again for the \(\Lambda\) system) as a two component object \( \left( \alpha, \gamma\right)^T \) , where the first component  is in turn a vector in \( \mathbb{C}^2 \). It follows that, in this case, \( \omega\in M_2(\mathbb{C}) \) (in fact, hermitian), \( \Delta\in M_1(\mathbb{C}) \) and hermitian, so in fact \(\Delta\in\mathbb{R}\) for the \(\Lambda\) system case. Finally, \(\Omega\) is a row two component vector for the \(\Lambda\) system. 

For the general case, consider that the full Hamiltonian of system  (\ref{eq:gensystem})  acts on the Hilbert space \( \mathcal{H} \). Then the projectors \( P \) and \( Q \) give the Hilbert subspaces \( P \mathcal{H} \) and \( Q \mathcal{H} \) respectively, with \( P \mathcal{H}\oplus Q \mathcal{H}= \mathcal{H} \). We have introduce new notation  in (\ref{eq:newnotation}) for the operators \( P_iHP_j \), where \( P_i \) stand for either \( P \) or \( Q \). Namely, \( PHP\to \omega\in \mathcal{B}(P \mathcal{H}) \), \( QHQ\to \Delta\in  \mathcal{B}(Q \mathcal{H}) \), while \( QHP\to \Omega\in \mathcal{B}(P \mathcal{H},Q \mathcal{H}) \). Now, \(\alpha\) stands for an element of \( P \mathcal{H} \), while \( \gamma\in Q \mathcal{H} \).

As we are looking for linearly invariant subspaces, let us define them by the relation \( \gamma= B \alpha \). Notice that the embedding operator \( B \) belongs to  \(  \mathcal{B}(P \mathcal{H},Q \mathcal{H}) \), as \(\Omega\). The invariance condition or embedding equation then reads
\begin{equation}\label{eq:blochexplicit}
\Omega+ \Delta B= B \omega+ B \Omega^{\dag} B\,.
\end{equation}
By itself, this equation, known in the literature as {\it Bloch's equation} \cite{0305-4470-36-20-201}, is not advantageous with respect to the direct analysis of the Hamiltonian. However, in the case of interest to us, the spectrum of \(\Delta\) is energetically very separated from the spectrum of \(\omega\), and we shall carry out a perturbative expansion or an iterative procedure to determine \( B \). Once \( B \) has been obtained, the evolution in the subspace is determined by \( \I\partial_t \alpha=\left( \omega+ \Omega^{\dag}B\right) \alpha \). This also provides us with approximate evolutions in the subspace, given an approximate solution to Bloch's equation. Notice the existence of an, in general, non-hermitian linear operator \( h_{ \mathrm{eff}}=\omega+ \Omega^{\dag}B \), which plays the role of an effective Hamiltonian.

\section{Expansion beyond adiabatic elimination}

Let us assume that \(\Delta\) has a bounded inverse, and that, for some definition of the operator norm, \( \epsilon= \| \Delta^{-1}\| \| \omega\|\ll 1 \) and \( \epsilon'=  \| \Delta^{-1}\| \| \Omega\|\ll 1\). We define the nonlinear transformation of operators
\begin{equation}
T(A)= - \Delta^{-1} \Omega + \Delta^{-1} A \omega+ \Delta^{-1} A \Omega^{\dag} A\,.
\end{equation}
The invariance condition (\ref{eq:blochexplicit}) may be now written as a fixed point equation, namely \( B=T(B) \). We shall now prove that the nonlinear transformation \( T \) has a fixed point which is the required solution to the problem at hand. Observe that, due to the definition of \( T \) and the properties of operator norms, we have
\begin{eqnarray}\label{eq:ineq}
\|T(A)\|&\leq& \|\Delta^{-1}\|\,\| \Omega\| + \|\Delta^{-1}\|\,\|  \omega\|\,\| A\| + \|\Delta^{-1} \|\,\| \Omega\|\,\| A\|^2 \nonumber\\
&\leq& \epsilon'\left(1+  \|A\|^2\right) + \epsilon \|A\|\,.
\end{eqnarray}
From this observation we obtain the following central proposition:
\begin{theorem}
Assume \( \epsilon, \epsilon'\geq0 \) and \(  \epsilon'\leq(1- \epsilon)/2 \). Let us define
\begin{equation}
r( \epsilon, \epsilon')= \frac{1- \epsilon}{2 \epsilon'}+ \sqrt{\left(\frac{1- \epsilon}{2 \epsilon'}\right)^2-1}\,.
\end{equation}
Then, the fixed point equation \( T(A)=A \) has at least one solution \( A_* \) such that 
\begin{equation}
\|A_*\|\leq r( \epsilon, \epsilon')\,.
\end{equation}
\end{theorem}
\begin{proof}
By direct analysis of the function \( g(x)= \epsilon'(1+x^2)+ \epsilon x \), and using inequality (\ref{eq:ineq}) we conclude that, for every \( A \) such that
\begin{equation}
\|A\|\leq r( \epsilon, \epsilon')
\end{equation}
it holds that 
\begin{equation}
\|T(A)\|\leq r( \epsilon, \epsilon')\,.
\end{equation}
As \( T \) maps a bounded closed convex set of the corresponding Banach space of operators into itself, there exists at least one fixed point  \( A_* =T(A_*)\)  in the set, by Schauder's fixed point theorem.

Furthermore, choosing the operator \( -\Delta^{-1} \Omega \) as the initial point of the iteration, one can readily see that 
\begin{equation}
\| -\Delta^{-1} \Omega\|\leq  \|\Delta^{-1}\|\,\| \Omega\| = \epsilon'\leq r( \epsilon, \epsilon')
\end{equation}
under the stated conditions. Thus, we know that there is at least one fixed point in its vicinity. 
\end{proof}

To sum up, under those conditions, we are assured of the existence of a solution of Bloch's equation which is small in the sense that \( \|B\|\leq r( \epsilon, \epsilon') \). It is therefore natural to attempt either an iterative process or a perturbative expansion to compute approximations for that solution. 

More concretely, we define
\begin{eqnarray}\label{eq:brecurrence}
B^{(0)}&=& - \Delta^{-1} \Omega\,,\\
B^{(k+1)}&=& T\left[B^{(k)}\right]\nonumber\,.
\end{eqnarray}
The sequence of operators \( B^{(k)} \) is thus defined by iteration, and lies in the region of applicability of the proposition. Although we have not proven the convergence of this sequence, the mapping of problem (\ref{eq:blochexplicit}) to the Schrieffer--Wolff expansion presented in Sec. \ref{sec:SW} allows us to conclude the uniqueness of the fixed point in some circumstances, given the results presented in \cite{Bravyi20112793}.

Alternatively, we can define a perturbative expansion as
\begin{equation}\label{eq:adiaelimexp}
B= \sum_{k=1}^\infty B_{(k)}
\end{equation}
and insert it into the embedding or Bloch  equation  to obtain the recursive equation
\begin{equation}\label{eq:adiaelimexprecurr}
B_{(k+1)}= \Delta^{-1} B_{(k)} \omega + \Delta^{-1} \sum_{l=1}^{k-1}B_{(k-l)} \Omega^{\dag} B_{(l)}\,,
\end{equation}
with the initial conditions $B_{(1)}=- \Delta^{-1} \Omega$ and $B_{(2)}= - \Delta^{-2} \Omega \omega$. Notice that, formally,
\[B^{(k)}-\sum_{l=1}^{k+1}B_{(l)}=O\left( \Delta^{-(k+2)}\right)\,,\]
which relates the perturbative and the recurrence results to a given order. Explicit formulae relating the iteration approximations \( B^{(k)} \) and the perturbative contributions \( B_{(k)} \) are hardly illuminating beyond this relation to a given order.

\section{The effective Hamiltonian}

As pointed out above, the evolution in the subspace fixed by \( B \) is determined by \( \omega+ \Omega^{\dag}B \). In the first order, this provides us with a hermitian operator
\begin{equation}
h_{(1)}=h^{(1)}= \omega - \Omega^{\dag} \Delta^{-1} \Omega\,,
\end{equation}
either by iteration or perturbatively. This is, in fact, what we would obtain from direct adiabatic elimination, namely, by setting \( \partial_t\gamma \) to zero, solving \( \gamma \) as \( - \Delta^{-1} \Omega \alpha \), and replacing this for \(\alpha\) in the differential equation. Furthermore, this operator is hermitian by construction.

Hermiticity, however, is not maintained in higher recurrence or perturbation orders~\cite{Muga2004}. This has been a sticking point in the literature, and the source of some confusion. It is clear that, by construction, the linear generator of evolution for the \(\alpha\) part of the full \(\psi=\left( \alpha, \gamma\right)^T\) does not need to be hermitian, even if the total Hamiltonian is hermitian. In any case, its spectrum must be real. In order to prove it, let us consider an eigenvector of the total Hamiltonian belonging to the linear subspace determined by \( \gamma=B \alpha \). Then, the restriction to its \(\alpha\) part will be an eigenvector of the effective Hamiltonian with the same eigenvalue as that for the total Hamiltonian. Conversely, let us assume that \(\alpha_*\) is an eigenstate of the effective Hamiltonian. Then, \( \psi_*=\left( \alpha_*, B \alpha_*\right)^T\) is also an eigenstate of the total Hamiltonian with the same eigenvalue, and hence the latter is real. 

In the finite-dimensional case of interest, it follows that the full effective Hamiltonian must be similar to a hermitian Hamiltonian. Let us construct this similarity transformation assuming that we already have a solution \( B \) for Eq. (\ref{eq:blochexplicit}) at hand. As a first step, notice that the total conserved norm \( \langle\psi,\psi\rangle \) can be expressed as
\begin{eqnarray*}
\langle\psi,\psi\rangle&=&\langle \alpha, \alpha\rangle+ \langle \gamma, \gamma\rangle\\
&=& \langle \alpha,\left(1+B^{\dag}B\right) \alpha\rangle\,.
\end{eqnarray*}
One is immediately led to examine \(h_1= \left(1+B^{\dag}B\right) \left(\omega+ \Omega^{\dag}B\right) \). By using Eq. (\ref{eq:blochexplicit}) and its conjugate, one concludes that \( h_1 \) is indeed hermitian, if \( B \) is a solution to Bloch's equation. Its spectrum, however, is not the one corresponding to the time evolution in the low energy sector \( P \mathcal{H} \). Let us define
\begin{equation}\label{eq:sb}
S_B= \sqrt{1+ B^{\dag}B}\,,
\end{equation}
which is always possible, since \( 1+ B^{\dag}B \) is positive. Then, for any constant unitary \( V \) acting on \( P \mathcal{H} \), we obtain a hermitian Hamiltonian 
\begin{equation}
h_V=V S_B \left( \omega+ \Omega^{\dag} B\right) S_B^{-1}V^{\dag}\,.
\end{equation}
Hermiticity is easily proven by noticing that
\begin{eqnarray*}
h_V&=&V S_B \left( \omega+ \Omega^{\dag} B\right) S_B^{-1}V^{\dag}\\
&=&V S_B^{-1}h_1S_B^{-1}V^{\dag}.
\end{eqnarray*}
The unitary \( V \) can be subsumed in the choice of the square root defining \( S_B \) in (\ref{eq:sb}). In what follows, the omission of the subscript \( {}_V \)  is associated to the assumption of a choice for \( S_B \).

In point of fact, an exact solution for (\ref{eq:blochexplicit}) is as hard to come by as an exact diagonalisation of the initial Hamiltonian, so we have to use approximate methods. However, if we truncated \( B \) to some approximation, \( h_1 \) (and thus \( h_V \)) would no longer be automatically hermitian. This can be mended by introducing in
\begin{equation}\label{eq:hermitian}
h= \left(1+B^{\dag}B\right)^{-1/2}\left( \omega + \Omega^{\dag}B+ B^{\dag} \Omega+ B^{\dag}\Delta B\right)\left(1+B^{\dag}B\right)^{-1/2}\,
\end{equation}
an approximate solution \( B_a \) to Bloch's equation. This expression would be equivalent (up to unitary transformations \( V \)) to \( h_V \), if \( B_a \) were an exact solution of Bloch's equation (\ref{eq:blochexplicit}); on the other hand, it is explicitly hermitian for any \( B_a \), which provides us with a hermitian approximate Hamiltonian encoding the effective evolution.
In the case of a perturbative expansion, one can compute (\ref{eq:hermitian}) to second order, for example, giving
\begin{equation}\label{eq:blochhamthird}
h^{(2)}= \omega - \Omega^{\dag} \Delta^{-1} \Omega- \frac{1}{2}\left( \Omega^{\dag} \Delta^{-2} \Omega \omega + \omega \Omega^{\dag} \Delta^{-2} \Omega\right)\,.
\end{equation}

\section{Relation to Schrieffer--Wolff expansion}\label{sec:SW}

In 1966, J. R. Schrieffer and P. A. Wolff introduced a perturbatively-built canonical transformation, with an antihermitian generator \( S \), in order to eliminate small denominators in the perturbative expansion of the Anderson Hamiltonian \cite{SchriefferWolff1966}. However, this technique is nowadays a staple in condensed matter physics.

Let us define an antihermitian operator \( S \) by
\begin{equation}\label{eq:swoper}
\tanh(S)= \begin{pmatrix}
0&B^{\dag}\\ -B&0
\end{pmatrix}\,,
\end{equation}
where \( S \) is called Schrieffer--Wolff's operator (see, for instance, \cite{Bravyi20112793}). Then, Bloch's equation (\ref{eq:blochexplicit}) and its conjugate can be written together in the form
\begin{equation*}
\begin{pmatrix}
0& \Omega^{\dag}\\ \Omega&0
\end{pmatrix}- \tanh(S)\begin{pmatrix}
0& \Omega^{\dag}\\ \Omega&0
\end{pmatrix}\tanh(S)= \left[\begin{pmatrix}
\omega&0\\0& \Delta
\end{pmatrix},\tanh(S)\right]\,.
\end{equation*}
Under the assumption that \( \cosh(S) \), with antihermitian and block anti-diagonal \( S \)), exists and is invertible, this equation is equivalent to \(\exp(S)H\exp(-S)\) being block-diagonal.

Some further formal manipulations give us the additional identification
\begin{eqnarray}
e^{-S}&=& \left[1-\tanh(S)\right]\left[1-\tanh^2(S)\right]^{-1/2} \nonumber \\
&=&  \begin{pmatrix}
1&-B^{\dag} \\ B& 1
\end{pmatrix} \begin{pmatrix}
\left(1+ B^{\dag}B\right)^{-1/2}&0\\0&\left(1+ BB^{\dag}\right)^{-1/2}
\end{pmatrix}\,.
\end{eqnarray}
This expression, together with the equivalence between Bloch's equation and that  \(\exp(S)H\exp(-S)\) is block diagonal, give us the equivalence between the two methods. The hermitian effective Hamiltonian \( h_V \) is, in fact, the block obtained from the Schrieffer--Wolff method corresponding to the low energy sector, and an approximate \( S \) gives us an approximate \( B \), and viceversa.

In the Schrieffer--Wolff method, the approximations are built by expanding \( S \) in powers of the interaction, and imposing that, to that order, the transformed Hamiltonian must be block-diagonal. To establish the connection between both methods, let us denote \( S= \begin{pmatrix}
0&- \mathcal{S}^{\dag}\\ \mathcal{S}&0
\end{pmatrix} \). Then, by using our notation, the first order Schrieffer--Wolff condition reads 
\begin{equation}\label{eq:SWfirst}
\mathcal{S}_0 \omega- \Delta \mathcal{S}_0= \Omega\,.
\end{equation}
Notice that this is a Sylvester equation, which will have a unique solution if the spectra of \(\omega\) and \(\Delta\) are disjoint. It should be pointed out that, in the context we are interested in, it might be necessary to compute approximate solutions for Eq. (\ref{eq:SWfirst}). Alternatively, if a solution for (\ref{eq:SWfirst}) can be computed, it will be a \emph{resummation} of our perturbative or iterative expansion.

To the same expansion order, the effective Hamiltonian for the low energy sector reads
\begin{equation}\label{eq:SWhamfirst}
h_{SW}^{(1)}= \omega+ \frac{1}{2}\left( \mathcal{S}_0^{\dag} \Omega+ \Omega^{\dag} \mathcal{S}_0\right)\,.
\end{equation}

\section{The $\Lambda$-system example}

The \(\Lambda\)-system in quantum optics described in the system of Eqs. (\ref{eq:brionlambda}) is amenable to both the approach based on  Bloch's equation and Schrieffer--Wolff's approach if \( \Delta\gg \delta,\tilde{ \Omega}_i \). Direct application of expressions (\ref{eq:blochhamthird}) and (\ref{eq:SWhamfirst}) yields
\begin{eqnarray}
h_{SW}^{(1)}&=& - \frac{ \Delta}{2} \sigma^z- \frac{1}{ \Delta} \frac{1}{ 1-( \delta/2 \Delta)^2} \tilde{\Omega}^{\dag} \tilde{\Omega}\\
&&+ \frac{ \delta}{ 4 \Delta^2}\frac{1}{ 1-( \delta/2 \Delta)^2} \left( \sigma^z \tilde{\Omega}^{\dag} \tilde{\Omega} + \tilde{\Omega}^{\dag} \tilde{\Omega} \sigma^z\right) \, , \nonumber \\
h^{(2)}&=& - \frac{ \Delta}{2} \sigma^z- \frac{1}{ \Delta}  \tilde{\Omega}^{\dag} \tilde{\Omega}\\
&&+ \frac{ \delta}{ 4 \Delta^2} \left( \sigma^z \tilde{\Omega}^{\dag} \tilde{\Omega} + \tilde{\Omega}^{\dag} \tilde{\Omega} \sigma^z\right) \nonumber \, .
\end{eqnarray}
In this simple example one can already see some features of the  Schrieffer--Wolff expansion as compared to the Bloch expansion; namely that the coefficients are not purely perturbative, but involve a resummation of perturbative terms. Additionally, this resummation presents with a pole that is not seen in the first perturbation terms. One should not expect the location of the pole to this order to be exact, and, in fact, working out specific exactly solvable examples (such as \( \tilde{\Omega}\to(1,0) \)) it is easy to see that it is not located at \( \delta=2 \Delta \)¥

\begin{figure}[h]
\begin{center}
\includegraphics[width=4in]{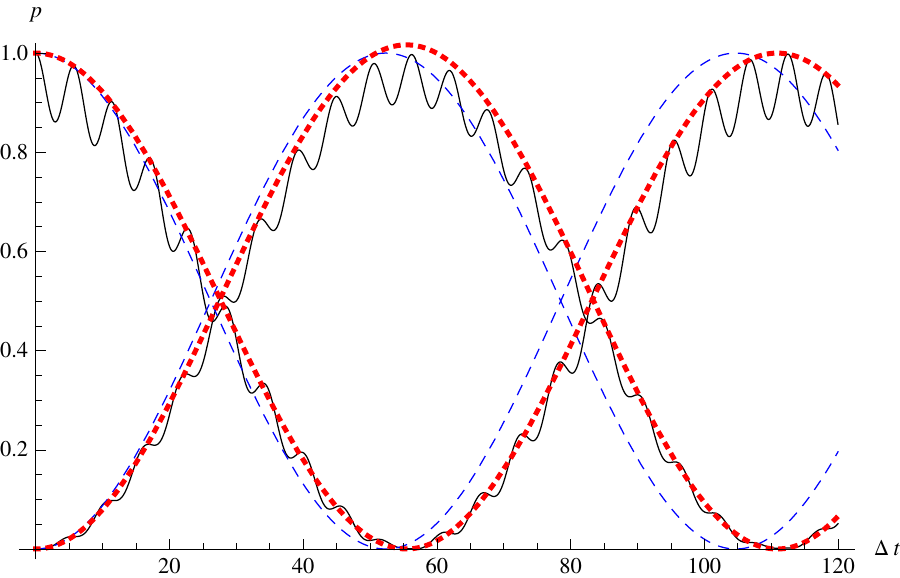}

\caption{Evolution of the population of the ground and excited states, with initial state \( (1,0,0)^T \), under a) the exact Hamiltonian (continuous black line), b) zeroth order effective Hamiltonian (dashed blue line) and c) fourth iteration of \( T \) (dotted red line). The parameters are \( \delta= -0.0175 \Delta\), \( \tilde{\Omega}_a=0.4 \Delta \), \( \tilde{\Omega}_b=0.3 \Delta \), for direct comparison with Ref. \cite{refId0}.}
\label{fig:lambdaex}
\end{center}
\end{figure}

Alternatively, the iteration expansion readily lends itself to numerical implementations. In Fig. \ref{fig:lambdaex} and \ref{fig:lambdaherm}, we show the presence of a secular shift for the adiabatic elimination approximation, the lack of normalisation with the non-hermitian Hamiltonian to some orders, and that the numerical solution of Bloch's equation matches perfectly to the real one already at fourth order of iteration, when ignoring large frequency oscillations.

\begin{figure}[h]
\begin{center}
\includegraphics[width=4in]{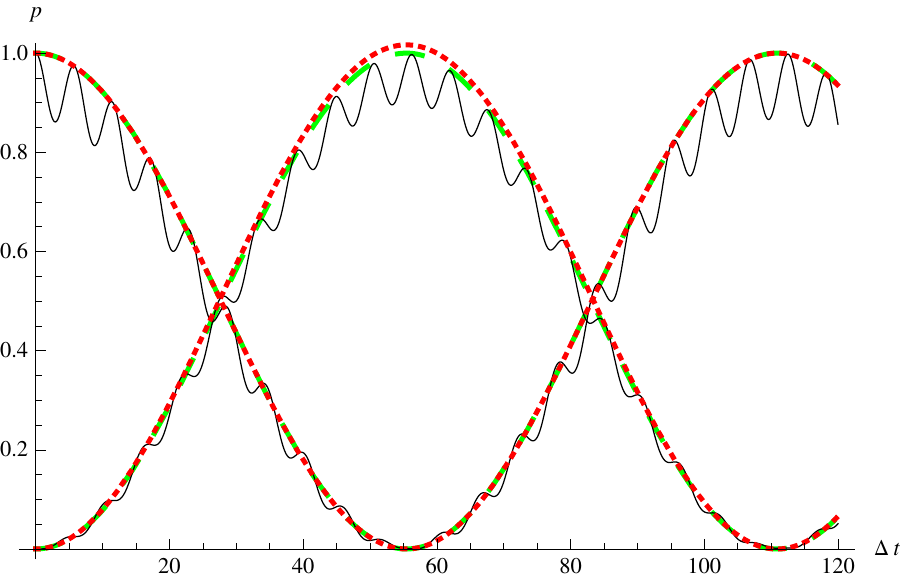}

\caption{Evolution of the population of the ground and excited states, with initial state \( (1,0,0)^T \), under a) the exact Hamiltonian (continuous black line), b) \( h_{ \mathrm{eff}}^{(4)} \) (dotted red line) and c) \( h_V^{(10)} \) with \( V=1 \) (dashed green line). The parameters are \( \delta= -0.0175 \Delta\), \( \tilde{\Omega}_a=0.4 \Delta \), \( \tilde{\Omega}_b=0.3 \Delta \), as before.}
\label{fig:lambdaherm}
\end{center}
\end{figure}

In particular, in Fig. \ref{fig:lambdaex}, we show the exact evolution of the populations of the lowest lying states in the \(\Lambda\) system. In the context of atomic physics, ``populations" translates into the square modulus of the coefficients of the state in an orthogonal  basis.

Thus, for a state \( \psi=\left( \alpha, \beta, \gamma\right)^T \) we are depicting \( | \alpha|^2 \) and \( | \beta|^2 \) as a function of normalised time \( \Delta t \), and these two quantities are compared with the evolution of the same populations with the first order of Bloch's approximation, i.e. adiabatic elimination approximation, and with the fourth iteration of the recurrence. One should observe the secular shift in the adiabatic elimination: the maxima of the relevant population under the evolution dictated by the adiabatic elimination are recurrently advanced with respect to the maxima given by the exact evolution. We also depict evolution under a non-hermitian effective Hamiltonian, and its effect is reflected in the fact that population maxima can be larger than one, as shown in the central maximum depicted in red. Since we are looking at the low energy effective evolution, the fast oscillations do not appear in the evolution under the effective Hamiltonian.

Analogously, in Fig. \ref{fig:lambdaherm}, one notices that the hermitian Hamiltonian to 10th order (obtained from computing the perturbative expansion of \( B \) to 10th order and substituting in (\ref{eq:hermitian})) perfectly matches the evolution of the exact populations, when one discards the rapid fluctuations that populate temporarily the highest energy state.

\section{Large frequency expansions for periodic Hamiltonians}

Let us consider a time-dependent periodic Hamiltonian with period \( 2\pi/ \omega \), so that \( H(t)=H(t+2\pi/\omega) \).  Let \( U(t) \) be the operator solution of \( i \partial_tU(t)=H(t)U(t) \), with \( U(0) = \mathbbm{1} \). Then, by Floquet's theorem, the eigenvalues of \( U(2\pi/\omega) \), when written as \( \exp\left(-i 2\pi \epsilon/\omega\right)  \), provide us with the so-called \emph{quasi-energies}  \(\epsilon\). It is well-known that the quasi-energies can be identified with the eigenvalues of the corresponding Floquet's Hamiltonian \cite{PhysRev.138.B979}. In order to understand it, let us assume that the original periodic Hamiltonian can be expanded as
\begin{equation}
H(t)=\sum_k H_k e^{-i k \omega t}
\end{equation}
with integer \( k \) and \( H_k^{\dag}=H_{-k} \).  Let us define the operators \( K \) and \( T \) acting on a different Hilbert space (the canonical example is \( L^2\left(S^1\right) \)), and such that a)  \(K  \) is hermitian and its spectrum consists of all integer numbers, i.e. if we denote the eigenbasis of \( K \) as \( |k\rangle \), then \( K|k\rangle=k|k\rangle \); b) \( T \) is unitary; c) \( \left[K,T\right]=T \). The Floquet's Hamiltonian associated to the original one is hence
\begin{equation}
H_F= \sum_k H_kT^k- \omega K\,.
\end{equation}
In order to avoid cluttering formulae unnecessarily we frequently omit the tensor product sign; by \( H_kT^k \) we actually denote \( H_k\otimes T^k \), which acts on \( \mathcal{H}\otimes L^2(S^1) \) if \( \mathcal{H} \) is the original Hilbert space on which \( H(t) \) acts. Similarly, \( K  \) stands for \( 1\otimes K \).

We shall now apply the formalism of adiabatic expansions to this Floquet's Hamiltonian to obtain approximate values for the quasi-energies.
The relevant projectors are \( P=\mathbbm{1}\otimes|0\rangle\langle0| \) and \( Q=\mathbbm{1}-P \). Therefore,
\begin{eqnarray}
PH_FP&=& H_0\,,\\
PH_FQ&=&\sum_k H_k PT^kQ= \sum_{k\neq0} H_k \otimes|0\rangle\langle -k|\,,\\
QH_FP&=& \sum_k H_k QT^kP= \sum_{k\neq0} H_k \otimes|k\rangle\langle 0|\,,\\
QH_FQ&=&-  \omega Q KQ+ \sum_k H_k QT^kQ\,.
\end{eqnarray}
The problem in the direct application of the formalism presented above lies in the computation of the inverse restricted to the \( Q \) subspace. However, in the limit of \( \omega \) much larger than \( H_k \) in norm, the inverse may be approximated by 
\begin{eqnarray}
Q(QH_FQ)^{-1}Q&=& - \frac{1}{\omega}\left(1- \frac{1}{ \omega}  K_Q^{-1}\sum_k H_k QT^kQ\right)^{-1} K_Q^{-1} \nonumber \\
&=& -\sum_{l=0}^{\infty} \frac{1}{(\omega)^{l+1}}\left[ K_Q^{-1}\sum_k H_k QT^kQ\right]^lK_Q^{-1}\,.
\end{eqnarray}
where \( K_Q^{-1} \) is the inverse of \( K \) restricted to the complement of \( |0\rangle \). By using now this approximation and truncating, we obtain a first approximation of the effective hermitian Hamiltonian in the constant sector as
\begin{equation}
H_0- \frac{1}{  \omega} \sum_{k\neq0} \frac{1}{k} H_{-k}H_k\,.
\end{equation}
As an example admitting an exact solution, let us consider
\begin{equation}\label{eq:exactham}
H(t)= g\left( \sigma^+ e^{i\omega t}+ \sigma^- e^{-i\omega t}\right)\,.
\end{equation}
The effective Hamiltonian reads
\begin{equation}\label{eq:examRWAres}
H_{ \mathrm{eff}}=  - \frac{g^2}{ \omega}\left(1-\frac{  g^2}{ \omega^2}\right) \sigma^z\,,
\end{equation}
up to \( O(g^6/\omega^5) \). It should be stressed at this point that the crucial information carried by the effective Hamiltonian is the set of eigenvalues. As stated, this example is exactly solvable, but it is also an interesting toy model to test the power of the techniques shown in this paper. To solve it, it is sensible to consider the evolution of a time-dependent Hamiltonian, which is, in physics language, the interaction picture Hamiltonian obtained from \( H_S=( \Delta+ \omega) \sigma^z/2+ g \sigma^x \) with respect to a free part given by \(H_{ \mathrm{free}}= \omega \sigma^z/2 \). For the sake of clarity, let us define precisely what the interaction picture is. Let \( U_S(t) \) be the unitary solution of the initial value problem \( i\partial_t U_S= H_S U_S \) with \( U_S(0)=1 \).  Similarly, let \( U_{ \mathrm{free}}(t) \) be the unitary solution of the initial value problem \( i\partial_t U_{ \mathrm{free}}=H_{ \mathrm{free}}U_{ \mathrm{free}} \) with \( U_{ \mathrm{free}}(0)=1 \). The interaction picture Hamiltonian with respect to the free part \( H_{ \mathrm{free}} \) is defined as
\begin{equation*}
H_I(t)= U_{ \mathrm{free}}(t)^{\dag} \left(H_S- H_{ \mathrm{free}}\right) U_{ \mathrm{free}}\,.
\end{equation*}
Now, let us consider the initial value problem \( i\partial_t U_I(t)=H_I(t)U_I(t) \), with \( U_I(0)=1 \) and \( U_I(t) \) unitary. Then one readily sees that the solution for this initial value problem can be written as
\begin{equation}\label{eq:interactprop}
U_I(t)= U_{ \mathrm{free}}(t)^{\dag}U_S(t)\,.
\end{equation}
Alternatively, if we are presented with a time dependent Hamiltonian, the corresponding evolution operator can be computed if the Hamiltonian is identified as the interaction picture Hamiltonian with respect to some free part.

For the specific \( H_S  \) and \( H_{ \mathrm{free}} \) above we have
\begin{eqnarray*}
H_I(t)&=& e^{i\omega t \sigma^z/2}\left( \frac{\Delta}{2} \sigma^z+ g \sigma^x\right)e^{i\omega t \sigma^z/2}\\
&=& \frac{\Delta}{2} \sigma^z+ g \left( \cos( \omega t) \sigma^x - \sin( \omega t) \sigma^y\right)\\
&=&  \frac{\Delta}{2} \sigma^z+ g \left( \sigma^+ e^{i \omega t}+ \sigma^- e^{-i \omega t}\right)\,.
\end{eqnarray*}

 Obviously, the time dependent Hamiltonian \( H(t) \) of (\ref{eq:exactham})  corresponds to this \( H_I(t) \) in the case \(\Delta=0\). Particularising to the example the solution (\ref{eq:interactprop}) we have
\begin{equation}
U_I(t)= \exp\left( \frac{i \omega t}{2} \sigma^z\right) \exp\left( -\frac{i(\omega+ \Delta) t}{2} \sigma^z- i g t \sigma^x\right) \,.
\end{equation}
Hence, by computing \( U_I(2\pi/\omega) \), its eigenvalues read \( (2\pi/\omega)\times\left(\omega/2\pm \sqrt{( \omega+ \Delta)^2/4+g^2}\right) \). Thus, the quasi-energies, which are defined modulo \( \omega \), can be fixed as
\begin{equation}
\mp\left[ \frac{ \omega}{2}- \sqrt{ \left( \frac{ \omega+ \Delta}{2}\right)^2+ g^2}\right]\,.
\end{equation}
If \( \omega \gg \Delta,g\), the quasi-energies may be expanded to obtain
\begin{equation*}
\pm\left( \frac{ \Delta}{2}+ \frac{g^2}{ \omega} - \frac{g^2 \Delta}{ \omega^2}+ \frac{g^2( \Delta^2-g^2)}{ \omega^3}+\cdots\right) \, ,
\end{equation*}
matching (\ref{eq:examRWAres}) in the case \(\Delta=0\).

One can also apply Schrieffer--Wolff's method in this case; no resummation is gained, however, since
 \( \mathcal{S}_0 \) is computed exactly as
 \begin{equation}
\mathcal{S}_0= - \frac{g}{ \omega}\left( \sigma^+\otimes|-1\rangle\langle0|- \sigma^- \otimes|1\rangle\langle0|\right)\, ,
 \end{equation}
giving \( h_{SW}^{(1)}= \frac{ g^2}{ \omega} \sigma^z \) for the resonant case \( \Delta=0  \). The case \(\Delta\neq0 \) can also be solved exactly for \( \mathcal{S}_0 \), resulting in 
 \begin{equation}
 \mathcal{S}_0=  \frac{2g}{ \Delta-2\omega}\left( \sigma^+\otimes|-1\rangle\langle0|- \sigma^- \otimes|1\rangle\langle0|\right)\,.
 \end{equation}
Therefore, in this case \(\Delta\neq 0\), the corresponding Hamiltonian is 
 \begin{equation}
 h_{SW}^{(1)}= \frac{ \Delta}{2} \sigma^z - \frac{2g^2}{ \Delta- 2 \omega} \sigma^z\,.
 \end{equation}
To summarise, we have shown that both Bloch's expansion and Schrieffer--Wolff's provide us with large frequency expansions for the quasi-energies in periodic Hamiltonians. As a further scope, it would be interesting to consider the application of these novel techniques to the quantum Rabi model beyond rotating-wave approximation~\cite{Braak2011,NoteToEditor}.

\section{Conclusions}

The well established adiabatic elimination procedure in quantum optics has produced some controversy in that scientific community on its meaning and on the feasibility of systematic improvements on that approximation. We have shown that, in fact, adiabatic elimination is the first term of a systematic expansion (be it perturbative or iterative) using Bloch's equation from nuclear optics. We have further shown the existence of relevant solutions. We next connected this approach to Schrieffer--Wolff's method, showing in which sense it can be said that Schrieffer--Wolff is a resummation. After the presented example, we have looked into the case of periodic Hamiltonians, employing the previous techniques on Floquet's Hamiltonian, to obtain high frequency expansions for quasi-energies. We expect these novel methods to be useful in current models of interest as is the case of the quantum Rabi model.

We acknowledge financial support from Basque Government Grants IT472-10 and IT559-10, UPV/EHU UFI 11/55, Spanish MINECO FIS2012-36673-C03-02 and FPA2009-10612, PROMISCE and SCALEQIT European projects.

 \bibliographystyle{spphys}

\begin{thebibliography}{10}
\providecommand{\url}[1]{{#1}}
\providecommand{\urlprefix}{URL }
\expandafter\ifx\csname urlstyle\endcsname\relax
  \providecommand{\doi}[1]{DOI \discretionary{}{}{}#1}\else
  \providecommand{\doi}{DOI \discretionary{}{}{}\begingroup
  \urlstyle{rm}\Url}\fi


\bibitem{walls2008quantum}
D.~Walls and G.~Milburn, \emph{Quantum Optics} (Springer, 2008).

\bibitem{Shore:2010uq}
B.~Shore, Acta Physica Slovaca, Reviews and Tutorials \textbf{58}(3), 243 (2010).

\bibitem{Yoo1985239}
H. I. Yoo and J.~Eberly, Phys. Rep. \textbf{118}(5), 239  (1985).

\bibitem{1751-8121-40-5-011}
E.~Brion, L. H. Pedersen, and K.~M{\o}lmer, J. Phys. A: Mathematical and Theoretical \textbf{40}(5), 1033 (2007).

\bibitem{refId0}
V. Paulisch, H. Rui, H. K. Ng, and B.-G. Englert, Eur. Phys. J. Plus \textbf{129}(1), 12 (2014).

\bibitem{Eden1955}
R. J. Eden and N. C. Francis, Phys. Rev. \textbf{97}, 1366 (1955).

\bibitem{Ellis1977}
P. J. Ellis and E.~Osnes, Rev. Mod. Phys. \textbf{49}, 777 (1977).

\bibitem{SchriefferWolff1966}
J.R. Schrieffer and P. A. Wolff, Phys. Rev. \textbf{149}, 491 (1966).

\bibitem{BornOppenheimer1927}
M.~Born and R.~Oppenheimer, Annalen der Physik \textbf{389}(20), 457 (1927).

\bibitem{Fenichel197953}
N.~Fenichel, Journal of Differential Equations \textbf{31}(1), 53  (1979).

\bibitem{0305-4470-36-20-201}
J.P. Killingbeck and G.~Jolicard, J. Phys. A: Mathematical and General
\textbf{36}(20), R105 (2003).

\bibitem{Bravyi20112793}
S.~Bravyi, D.P. DiVincenzo, and D.~Loss, Annals of Physics \textbf{326}(10), 2793 (2011).

\bibitem{Muga2004}
G. Muga, J. P. Palao, B. Navarro, and I. L. Egusquiza, Phys. Rep. {\bf 395}, 357 (2004).
\bibitem{PhysRev.138.B979}
J.H. Shirley, Phys. Rev. \textbf{138}, B979 (1965).

\bibitem{Braak2011}
D. Braak, Phys. Rev. Lett. {\bf 107}, 100401 (2011).

\bibitem{NoteToEditor}
We would like to cite here the contribution to this Proceedings of Daniel Braak.


\end{thebibliography}

\end{document}